\documentclass[lettersize,journal]{IEEEtran}
\usepackage{cite}
\usepackage{amsmath,amsfonts, amsthm}
\usepackage{array}
\usepackage[caption=false,font=normalsize,labelfont=sf,textfont=sf]{subfig}
\usepackage{textcomp}
\usepackage{stfloats}
\usepackage{url}
\usepackage{verbatim}
\usepackage{graphicx}

\usepackage{graphicx}
\usepackage{xcolor}

\usepackage{comment}
\usepackage{dsfont}
\usepackage{footmisc}
\usepackage{csquotes}
\usepackage{enumitem}
\usepackage{mathtools}
\usepackage{tikz}
\usepackage{algorithm}
\usepackage{algpseudocode}
\usepackage[short]{optidef}

\usepackage{caption, subcaption}

\usepackage{accents}
\usepackage[implicit=false, colorlinks=true, urlbordercolor=red, urlcolor = blue,]{hyperref}

\newcommand{\ie}{\textit{i.e., }}
\newcommand{\eg}{\textit{e.g., }}
    
\newcommand{\R}{\mathbb R}    
   
\newcommand\eqnumber{\addtocounter{equation}{1}\tag{\theequation}}

\newcommand{\E}{\mathbb{E}}

\newtheorem{theorem}{Theorem}

\newtheorem{remark}[theorem]{Remark}

\newcommand{\xjoint}{X} 
\newcommand{\ii}[2]{^{#1}_{#2}} 
\newcommand{\infoi}{X^{\text{obs}}} 
 
\newcommand{\barrierf}{\phi} 
\newcommand{\safeset}{\mathcal{C}} 

\newcommand{\numbarrierf}{B} 
\newcommand{\safeevent}{S} 
\newcommand{\unsafethreshold}{\epsilon}
\newcommand{\probof}{\mathbb{P}} 
\newcommand{\safeprob}{\mathbf{\Psi}} 
\newcommand{\algfunction}{\alpha} 
\newcommand{\expof}{\mathbb{E}} 
\newcommand{\availableinfo}{I} 

\newtheorem{definition}{Definition}

\newcommand{\rev}[1]{#1}            

\def\BibTeX{{\rm B\kern-.05em{\sc i\kern-.025em b}\kern-.08em
    T\kern-.1667em\lower.7ex\hbox{E}\kern-.125emX}}
\usepackage{balance}
\begin{document}
\title{Safe Driving in Occluded Environments}
\author{Zhuoyuan Wang$^{1}$, Tongyao Jia$^{1}$, Pharuj Rajborirug$^{1,2}$,
Neeraj Ramesh$^{1}$, Hiroyuki Okuda$^{3}$, Tatsuya Suzuki$^{3}$, Soummya Kar$^{1}$, Yorie Nakahira$^{1*}$
\thanks{$^{1}$Zhuoyuan Wang, Tongyao Jia, Pharuj Rajborirug, Neeraj Ramesh, Soummya Kar, and Yorie Nakahira are with the Department of Electrical and Computer Engineering, Carnegie Mellon University,
        {\tt\small \{zhuoyuaw,tongyaoj,prajbori,neerajr,soummyak, \\
        ynakahir\}@andrew.cmu.edu}.}%
\thanks{$^{2}$Pharuj Rajborirug is also with the Faculty of Medicine, King Mongkut's Institute of Technology Ladkrabang, Thailand,
        {\tt\small pharuj.ra@kmitl.ac.th}.}
\thanks{$^{3}$Hiroyuki Okuda and Tatsuya Suzuki are with the Department of Mechanical Systems Engineering, Nagoya University, Japan,
        {\tt\small \{h\_okuda, t\_suzuki\}@nuem.nagoya-u.ac.jp}.}%
\thanks{$*$To whom correspondence should be addressed.}}


\maketitle

\begin{abstract}
Ensuring safe autonomous driving in the presence of occlusions poses a significant challenge in its policy design. 
While existing model-driven control techniques based on set invariance can handle visible risks, occlusions create latent risks in which safety-critical states are not observable. Data-driven techniques also struggle to handle latent risks because direct mappings from risk-critical objects in sensor inputs to safe actions cannot be learned without visible risk-critical objects. 
Motivated by these challenges, in this paper, we propose a probabilistic safety certificate for latent risk. Our key technical enabler is the application of probabilistic invariance: It relaxes the strict observability requirements imposed by set-invariance methods that demand the knowledge of risk-critical states.
The proposed techniques provide linear action constraints that confine the latent risk probability within tolerance. Such constraints can be integrated into model predictive controllers or embedded in data-driven policies to mitigate latent risks. 
The proposed method is tested using the CARLA simulator and compared with a few existing techniques. The theoretical and empirical analysis jointly demonstrate that the proposed methods assure long-term safety in real-time control in occluded environments without being overly conservative and with transparency to exposed risks.

\end{abstract}

\begin{IEEEkeywords}
Autonomous driving, safe control, occlusions, latent risks. 
\end{IEEEkeywords}

\section{Introduction}
\label{sec:introduction}
\IEEEPARstart{V}{isual} 
occlusions impose significant challenges in the policy design of autonomous driving. Most sensors cannot see through opaque objects, and there can be large unobserved regions and various latent risks~\cite{poncelet2020safe, gilroy2019overcoming,gangadhar2023occlusion}. The stochastic nature of road users---such as other vehicles and pedestrians---further complicates the problem~\cite{zhang2021safe,kocc2021pedestrian}. Given such uncertainties, avoiding all latent risk objects in the worst case may not be feasible, or such policies can significantly compromise performance due to their overly conservative nature. Accounting for latent risks in the long term often requires computations that can be prohibitive for real-time control or onboard resources.  

    

\begin{figure}[t]
    \centering
    \includegraphics[width=0.9\linewidth]{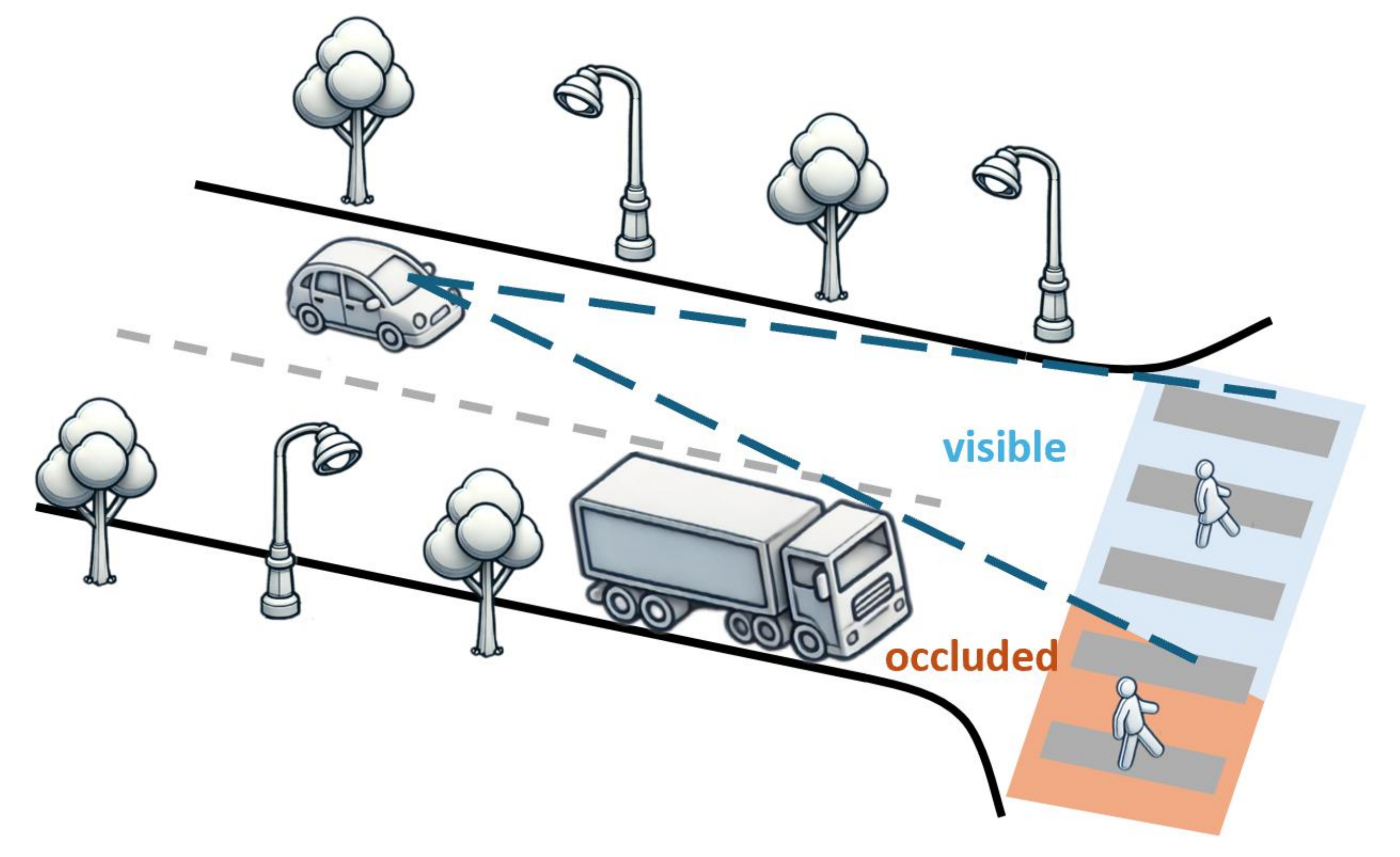}
    \vspace{-2mm}
    \caption{Visually occluded intersection scenario of interest.}
    \label{fig:intersection}
\end{figure}

In this paper, we study the safe navigation of autonomous vehicles in occluded environments. Examples of such scenarios are visually occluded intersections, illustrated in Fig.~\ref{fig:intersection}. We propose a probabilistic safety certificate that controls latent risk (risk imposed by potential road users and objects behind occlusions) within tolerance. The key technique used in the safety certificate is a novel notion of probabilistic invariance. It relaxes the required information in the set invariance approach~\cite{blanchini1999set} while preserving its computational efficiency. It also gives linear constraints on actions, which can be solved efficiently in optimization-based or model-predictive controllers. 
The technical merits of the proposed method are
\begin{enumerate}
    \item Guaranteed long-term assurance against latent risks arising from invisible or occluded objects (Theorem~\ref{thm:1}).
    \item Avoid over-conservatism and achieve better safety and performance trade-offs over existing methods (Fig. 5, Table~\ref{tab:emp_safe_prob}). 
    \item Transparency in design and to the exposed risks (see the proposed optimization-based controller~\eqref{eq:safe_control_optimization}, Fig.~\ref{fig:heat_map_psi}, Remark~\ref{rmk:safe prob meaning} and Remark~\ref{rmk:ease of design}).
\end{enumerate}

The rest of the paper is organized as follows. Section~\ref{sec:related works} discusses the related works. Section~\ref{sec:problem_formulation} formulates the safe control problem of interest. Section~\ref{sec:proposed method} introduces the proposed occlusion-aware control framework. Section~\ref{sec:experiments} presents CARLA simulations to validate the proposed method. Finally, Section~\ref{sec:conclusion} concludes the paper.

\section{Existing techniques for driving in occluded environments}
\label{sec:related works}

Many techniques have been developed for the detection of occlusions and risk-critical road users and occlusion detection (see~\cite{gilroy2019overcoming} for a review of these techniques). Some existing literature has also focused on avoiding occlusions by improving the transportation infrastructure (e.g., building a vehicle-to-infrastructure network) or finding a path with greater visibility (e.g., changing lanes to increase visibility~\cite{narksri2022occlusion}).

This paper focuses on decision-making for vehicles in occluded environments when occlusion may not be avoidable. Below we summarize existing approaches to design and learning decision policies.
Partially observable
Markov decision process (POMDP) is used in~\cite{schratter2019pedestrian,thornton2018autonomous} to model the vehicle control problem under occlusions, and an optimal policy is solved given an utility function to maximize. 
Similarly,~\cite{hubmann2019pomdp} characterizes the existence of latent risk under occlusion via belief states in a POMDP, and~\cite{zhang2021improved,zhao2024inference} incorporate 'phantom objects' that describe potential road users behind occlusions in the POMDP setting. Local path planning, trajectory planning, and speed planning are used to avoid collisions in the presence of occlusion~\cite{wang2023occlusion}. 
In~\cite{trauth2023toward}, the likelihood of occlusions and collision probabilities are taken into account by a sampling-based trajectory planner to minimize risk. 
In~\cite{yu2019occlusion}, the authors use road layouts to
forecast and quantify risk associated with occlusions, and incorporate low-level and high-level planning algorithms to generate safe trajectories. However, these techniques are often computationally expensive and require powerful onboard computers. On the other hand, stochastic model predictive control (MPC) is used in~\cite{li2024smpc, chen2023drive, maki2022stochastic} to enforce probabilistic safety constraints under occlusions, with explicitly modeled risk or pedestrian occurrence distribution online. Tube-based and occlusion-aware MPCs are used to ensure the safety and feasibility of ego vehicles under occlusion by computing reachable sets of controllers~\cite{firoozi2022occlusion, farrokhsiar2013integrated}. As these techniques only assure safety/performance during the MPC outlook time horizon, such techniques are subject to stringent tradeoffs between the time horizon to assure safety vs. computation burden (latency in real-time control). In contrast, the proposed approach leads to linear constraints and can be integrated into a quadratic program for efficient real-time control.

Many learning-based methods have also been developed for driving. Some techniques learn humans' risk perception which can be used for the design of driving algorithms. For example, Driver’s Risk Field (DRF) model, which represents a driver’s belief about event probabilities and perceived risk, was used to explain human-like driving behaviors~\cite{kolekar2020human}. Other techniques learn driving policies from data using imitation learning and reinforcement learning. Specifically, behavior cloning from human expert drivers via deep learning is studied in~\cite{sama2020extracting} to account for risk from occlusions. 
Deep reinforcement learning is used in~\cite{isele2018navigating,bouton2019safe,kamran2020risk} to solve for vehicle controls at occluded intersections. However, these methods require extensive data covering all possible situations. Moreover, it is challenging to establish a direct mapping from risk-critical objects to safe actions in data-driven techniques, the black-box nature makes it difficult to obtain safety guarantees. In comparison, the proposed method ensures that latent risk probability is controlled within tolerance.

\section{Problem Formulation}
\label{sec:problem_formulation}
In this section, we present the problem formulation, which includes the vehicle dynamics
in section~\ref{subsec:vehicle_model}, interaction model in section~\ref{subsec:interaction_model}, the occlusion model in~\ref{subsec:occlusion_model}, 
and the safety specifications in section~\ref{subsec:safety_specs}.

We consider a visually occluded intersection consisting the ego vehicle and crossing pedestrians.
Fig.~\ref{fig:intersection} visualizes the scenario. We model the overall road system by the following discrete-time dynamics:
\begin{align}
\label{eq:overall_dynamics}
    X^{\text{all}}_{k+1} = f^{\text{all}}(X^{\text{all}}_k) + g^{\text{all}}(X^{\text{all}}_k) \; U^{\text{all}}_k + \sigma(X^{\text{all}}_k) dW_k,
\end{align}
where $X^{\text{all}} \in \mathbb{R}^z$ is the state, and $U^{\text{all}} \in \mathbb{R}^m$ is the control input, $W \in \mathbb{R}^q$ is the standard Wiener process with $W_0 = \boldsymbol{0}$, $\sigma: \mathbb{R}^z \rightarrow \mathbb{R}^q$ is the magnitude of the noise, and $k$ is the time step. Here the state includes the vehicle states and the pedestrian states. We consider discrete time in this paper so that the proposed method can be implemented by a digital controller. One can discretize any continuous dynamics into the form of~\eqref{eq:overall_dynamics} as in~\cite{ogata1995discrete}.
Note that in this paper we focus on scenarios at intersections where the ego vehicle is interacting with pedestrians, but the formulation and the proposed method can be easily generalized to cases at different scenarios with other road users (\eg crossing vehicles) by replacing the vehicle, interaction, and occlusion models used. In other words, our proposed method is scenario agnostic.

In the following subsections, we will formally introduce the vehicle dynamics, the pedestrian models, the occlusion settings, and the safety specifications.

\subsection{Vehicle Dynamics and Nominal Control}
\label{subsec:vehicle_model}
We consider the following discrete-time control-affine dynamics for the ego vehicle: 
\begin{align}
\label{eq:dynamics}
    X^{\text{ego}}_{k+1} = f^{\text{ego}}(X^{\text{ego}}_k) + g^{\text{ego}}(X^{\text{ego}}_k) U^{\text{ego}}_k
\end{align}
where $X^{\text{ego}} \in \R^n$ is the vehicle's state, $U^{\text{ego}} \in \R^m$ is the control input to the vehicle, and $f^{\text{ego}}: \R^n \to \R^n$ and $g^{\text{ego}}: \R^n \to \R^{n \times m}$ describes the vehicle dynamics. Since the only control input to the system is the control input to the vehicle, we use $U$ instead of $U^{\text{ego}}$ in the rest of the paper for conciseness.
The choice of the vehicle model can range from simple double-integrators \cite{liang2000string} to complete 6 DoF models \cite{kiencke2000automotive}. 
The control input \(U\) is generated by a predefined control law \(N: \mathbb{R}^n \to \mathbb{R}^m\): 
\begin{align}
    U = N(X^{\text{obs}}),
\end{align}
where $X^{\text{obs}}$ is the observable states to the vehicle, \ie the available information.
This nominal controller is typically designed to meet performance criteria, such as tracking a planned trajectory, and it can be generated using either data-driven methods or optimization-based approaches (\eg MPC~\cite{kouvaritakis2016model}). 
However, it may not assure safety, particularly in the presence of latent risks. The closed-loop vehicle dynamics with the nominal controller is given by:
\begin{align}
\label{eq:closed-loop}
    X^{\text{ego}}_{k+1} = f^{\text{ego}}(X^{\text{ego}}_k) + g^{\text{ego}}(X^{\text{ego}}_k) N(X^{\text{obs}}_k).
\end{align}


\subsection{Interaction Model}
\label{subsec:interaction_model}
We represent pedestrian behavior as a combination of decision-making and motion dynamics. The decision-making component determines the agent’s high-level choices by considering the surrounding context. 
Let \(X^{\text{all}} \in \mathbb{R}^z\) denote the joint state of every agent involved in the interaction, and let \(X^{\text{ped}} \in \mathbb{R}^{z-n}\) refer to the pedestrian’s state.
We denote by \(\mathcal{Z}\) the external factors impacting decision-making—such as physical settings, social contexts, and traffic characteristics, as explained in \cite{rasouli2019autonomous}. A decision-making function then yields a distribution over potential intentions (\eg go/wait, lane-keep/lane-change) based on \(X^{\text{all}}\) and the context \(\mathcal{Z}\). Formally, the pedestrian’s decision-making process is given by
\begin{align}
d_k \sim \mathcal{D}(d_k \mid X^{\text{all}},\mathcal{Z}_k),
\label{eq:decision_make_ped}
\end{align}
where \(\mathcal{D}\) is the decision model’s distribution, \(d\) represents the pedestrian’s decision, and we assume that \(d\) may take values from a finite set. In practice, this process can be modeled as a finite state machine \cite{Kielar2014Concurrent}, an interactive multiple model (IMM) filter \cite{burger2020interaction}, a POMDP \cite{hubmann2018belief}, or a neural network \cite{Rasouli2017Are}.

Once a decision \(d\) is specified, the pedestrian dynamics dictate how the pedestrian moves in accordance with that intention. Specifically, these dynamics can be expressed as
\begin{align}
\label{eqn:interaction_update}
X^{\text{ped}}_{k+1} \sim P^{\text{ped}}(X^{\text{ped}}_{k+1} \mid X^{\text{ped}}_k, d_k),
\end{align}
where \(P^{\text{ped}}\) describes the distribution underlying the pedestrian’s state update function. Prior work has employed social force models \cite{Helbing1995Social} and recurrent neural networks \cite{Camara2020Pedestrian} to instantiate \(P^{\text{ped}}\). For any given decision \(d\), many approaches assume that each state in \(X^{\text{ped}}\) follows a Gaussian distribution—an assumption supported by the Central Limit Theorem, given the accumulation of various sources of noise and uncertainty~\cite{yoon2021interaction, ellis2009modelling}. In our framework, we assume that the dynamics in \eqref{eq:dynamics} and \eqref{eqn:interaction_update} result in \eqref{eq:overall_dynamics}. Detailed information about the interaction model utilized in our experiments is provided in Section~\ref{sec:experiments}.

\subsection{Occlusion Model}
\label{subsec:occlusion_model}
We define occlusion as any area that is outside the field of view of the ego vehicle's sensing modalities (\eg cameras, radar, sonar). 
Occlusion $\mathcal{H}_k$ is defined in a map space $\mathcal{M}$, where $\mathcal{O}_k$ is the occupied space by objects and $\mathcal{V}(X^{\text{ego}}_k, \mathcal{O}_k)$ is the visible space in the field of view (FOV) of the ego vehicle, all at time $k$.
The occlusion \(\mathcal{H}_k\) is then defined by
\begin{align}
\mathcal{H}_{k} = 
(\bar{\mathcal{O}}_k \cap \bar{\mathcal{V}}(X^{\text{ego}}_k, \mathcal{O}_k)) \in \mathcal{M},
\label{eq:def_occlusion}
\end{align}
where \(\bar{\mathcal{O}}\) and \(\bar{\mathcal{V}}\) refer to those parts of \(\mathcal{M}\) that are not within \(\mathcal{O}\) or \(\mathcal{V}\), respectively. Thus, \(\mathcal{H}_k\) represents the map area that is neither occupied by obstacles nor visible to the ego vehicle. The methods for estimating \(\mathcal{O}_k\) and \(\mathcal{V}(X^{\text{ego}}_k, \mathcal{O}_k)\) in \eqref{eq:def_occlusion} depend heavily on the sensors used. For instance, with LiDAR, one often utilizes neural networks \cite{lang2019pointpillars} to detect \(\mathcal{O}_k\), and then applies ray casting \cite{Zhang2019Lidar} to approximate \(\mathcal{V}(X^{\text{ego}}_k, \mathcal{O}_k)\). Although infrastructure-to-vehicle (I2V) or vehicle-to-vehicle (V2V) communications can mitigate occlusions in some scenarios \cite{muller2022motion}, they cannot resolve all occlusion conditions across diverse driving environments. While occlusion detection itself lies outside the scope of this paper, the proposed approach is capable of incorporating different sizes and shapes of any identified occlusion as parameters.

\subsection{Safety Specification}
\label{subsec:safety_specs}

Our goal is to ensure the long-term safety of all road users. It is assumed that there are $\numbarrierf$ safety specifications for the overall interaction system, indexed by $j \in \{1,2,\cdots,\numbarrierf\}$, and each specification is represented as follows: specification $j$ is defined by the event
\begin{align}
    \safeset_{j}=\{X^{\text{all}}\in\R^{z}:\barrierf_{j}(X^{\text{all}})\geq 0\},
\end{align}
where $\barrierf_j:\R^{z}\rightarrow \R$ is a continuous mapping. The definition can capture various safety requirements in autonomous driving, for example, all road users do not collide with each other, and the vehicle's speed should be less than a certain value when it is close to other vehicles.
Let
\begin{align}
    \safeevent=\{\xjoint\ii{}{\tau}\in\safeset_j,\forall\tau\in\{k,k+1,\cdots,k+T\}, \forall j\},
\end{align}
where $T$ is the outlook time horizon.
The goal is to ensure that
\begin{align}
    \label{eq:safety_goal}
    \probof(\safeevent)\geq 1-\unsafethreshold,\ \forall k\geq 0, 
\end{align}
where $\unsafethreshold$ is a design parameter chosen to specify the tolerable probability of risk. This parameter provides flexibility for system designers to balance conservativeness and operational feasibility based on application requirements. Due to occlusions, there may be cases when feasible solutions for safety with probability one do not exist---for example, the vehicles may not be able to move forward until all information of occluded area is obtained (Remark \ref{rm:safety-with-probability-one}). This formulation allows the system to use non-zero but sufficiently small $\unsafethreshold$ to improve feasibility in such situations.  

\begin{remark}
\label{rm:safety-with-probability-one}
For general stochastic differential equations, stochastic invariance gives conditions to remain within a known set at all times~\cite{zhou2023generalized,da2004invariance,doss1977liens,abi2018stochastic,aubin1990stochastic}. However, these conditions may not hold in certain driving environments when there are large uncertainties due to diverse factors such as occlusions and unobserved variables associated with agent motion (\eg managing tangential volatility and inward-pointing compensated drift~\cite{abi2018stochastic}). 

\end{remark}

In this paper, we focus on the case study of collision avoidance at visually occluded intersections. In this scenario, the safety specification is given by
\begin{equation}
    \safeset = \{X^{\text{all}}\in\R^{z}: \|p - p^{\text{ped}_i}\| \geq d_{\text{min}}, \forall i\},
\end{equation}
where $p$ is the position of the ego vehicle, $p^{\text{ped}_i}$ is the position of the $i$-th pedestrian, and $d_{\text{min}}$ is the required minimal distance between the vehicle and the pedestrians.

\section{Proposed Method}
\label{sec:proposed method}
In this section, we first present a safe condition to ensure the long-term safety of the system in section~\ref{subsec:safety_cond}, and then show how to quantify risk in occluded intersection in section~\ref{subsec:risk_estimation}. After that, we present our safe control algorithm in section~\ref{subsec:safe_control}.



\subsection{Condition for Assuring Safety}
\label{subsec:safety_cond}

In this subsection, we present a sufficient condition for the long-term safety specifications~\eqref{eq:safety_goal}. Let
\begin{align}
\label{eq:safe_prob_def}
    \safeprob(\availableinfo):=\probof(\safeevent|\availableinfo)\in\R
\end{align}
be the sequence of probability of event $\safeevent$ conditioned on the information $\availableinfo$. 

\begin{remark}
\label{rmk:safe prob meaning}
The variable $\mathbf{\Psi}(\availableinfo)$ has the physical meaning of the safety probability of the system in the long term. Its value at $\availableinfo$ indicates how risky the system will be in the future, evolving from a state with information $\availableinfo$. 
\end{remark}

We define a notion of conditional discrete-time generator as below. 

\begin{definition}
\label{def:afy}
(Conditional discrete-time generator).
The conditional discrete-time generator $A$ of a discrete-time stochastic process $\{x_k \}_{k\in\mathbb{Z}_+}$ conditioned on another process $\{y_k \}_{k\in\mathbb{Z}_+}$ with sampling interval $\Delta t$ evaluated at time $k$ is given by
    \begin{align}
    \label{eq:afy}
        A\phi(x_k|y_k)=\frac{\E[\phi(x_{k+1})| y_k]- \E[ \phi(x_k)| y_k ] }{\Delta t}
    \end{align}
    whose domain is the set of all functions $\phi:\R^n\rightarrow\R$ of the stochastic process\footnote{Note that while the discrete-time generator is generally defined over all functions $\phi$, in our case of safe control, the function of interest is the long-term safety probability $\safeprob$.}.
\end{definition}
When $x_k = y_k$, this generator becomes the discrete-time counterpart of the continuous-time infinitesimal generator. The conditioning of $y_k$ is to capture the ego vehicle's limited information due to occlusions.
Although the value of $A\phi(y_k)$ depends on both $x_k$ and $y_k$, with a slight abuse of notation, for the rest of the paper, we will use $A\phi(y_k)$ where the discrete-time stochastic process $x_k$ in Definition~\ref{def:afy} is the full state of the interaction system, \ie $X^{\text{all}}_k$ in 
\eqref{eq:overall_dynamics}.

Let $\infoi_{k}$ be the information that the ego vehicle can acquire at time $k$. Note that $\infoi_{k} = X^{\text{ego}}_k$ if no other road users appear from the occlusions.

We consider the following condition at all time $k$:
\begin{align*}
    A\safeprob(\infoi_{k})\geq -\algfunction(\safeprob(\infoi_{k})-(1-\unsafethreshold)),\ \forall k\geq 0. \label{eqn:safe_cond} \eqnumber
\end{align*}
Here, $\algfunction: \R \rightarrow \R$ is a function that satisfies the following 2 design requirements:
\begin{itemize}[leftmargin=*]
    \item[] Requirement 1: $\algfunction(h)$ is linear and increasing in $h$.
    \item[] Requirement 2: $\algfunction(h) \leq h$ for any $h \in\mathbb{R}^+$.
\end{itemize}
Note that $\alpha(h) = \eta h$ for all $\eta \in (0,1)$ satisfies the above requirements.\footnote{\rev{Such requirements are essential to provide safety guarantees in Theorem~\ref{thm:1}.}}
The probability measure of $\probof(\safeevent|\availableinfo)$ is taken over $X^{\text{all}}$, the global state, conditioned on $\infoi$, the information that can be accessed by the ego vehicle. Therefore, the values on both sides of \eqref{eqn:safe_cond} can be computed using $\infoi$. 
\rev{Intuitively, condition~\eqref{eqn:safe_cond} enforces that the gradient of the long-term safety probability $\Psi$ to be positive if its value drops below the desired threshold $1-\epsilon$, ensuring that the long-term safety probability always maintains above the threshold. In the following, we give a formal theorem for the long-term safety guarantee from condition~\eqref{eqn:safe_cond}.}

\begin{theorem}
\label{thm:1}
Consider systems \eqref{eq:dynamics} and \eqref{eqn:interaction_update} which forms~\eqref{eq:overall_dynamics}. We assume the initial condition $X^{\text{all}}_0=x^{\text{all}}_0$ satisfies $\probof(\safeevent|X^{\text{all}}_0=x^{\text{all}}_0)\geq 1-\unsafethreshold$. If at each time $k$, the ego vehicle generates a control policy that satisfies \eqref{eqn:safe_cond}, then the following condition holds:
\begin{align}
    \probof(\safeevent) = \expof[\probof(\safeevent|x^{\text{all}}_{k})]\geq 1-\unsafethreshold,\ \forall k\geq 0. \label{eq:safety_guarantee}
\end{align}
\end{theorem}

\begin{proof}
    See~\cite[Theorem 1]{jing2022probabilistic}, where the ego vehicle is the controlled agent, and the available information is $\infoi$.
\end{proof}
Theorem~\ref{thm:1} says the long-term safety of the system is guaranteed by the proposed safety condition~\eqref{eqn:safe_cond} for all time with desired probability. Note that the safety condition~\eqref{eqn:safe_cond} only involves the available information $\infoi$.

Condition~\eqref{eqn:safe_cond} is a condition on the control input as the infinitesimal generator gives the derivative of the safety probability over time with regard to certain control.
One can also solve for safe control through constrained optimizations incorporating condition~\eqref{eqn:safe_cond}. In practice, one might need to calculate the safety probability and its gradients numerically, as indicated in~\cite{gangadhar2023occlusion}. 
In this following sections we will show how to estimate the safety probability and find safe control for the collision avoidance problem at occluded intersection.

\subsection{Risk Estimation}
\label{subsec:risk_estimation}
In this section, we introduce how the probability of safety under visual occlusion is estimated. 
\rev{While we adopt a Monte Carlo-based method below, our framework can use other risk quantification or learning techniques to estimate the safe probability~\cite{wang2022potential, yu2019occlusion, katrakazas2019new, wang2026myopically, wang2023generalizable, wang2025generalizable, wang2025neural}. In the context of risk probability due to occlusion, efficient techniques for online computation have also been developed~\cite{yamada2023bayesian}.} 

For an ego vehicle starting at initial position and velocity $(p_0, v_0)$, we sample multiple trajectories to get an empirical estimate of the safety probability over certain horizon $T$. Specifically, at each trial we run a nominal controller when there is no pedestrian in sight, and apply emergency brakes when there are visible pedestrians, until time reaches $T$. 
We identify the safety of the ego vehicle by recording if any collision happens during the trial.
The detailed procedures for safety identification is summarized in Algorithm~\ref{alg:determine_safety}. 
Then, for the fixed time horizon $T$ and the initial state, we run the safety identification procedure for $N$ times and obtain the ratio of safe trajectories $F = \sum_{n=1}^N \frac{s_n}{N}$, as shown in Algorithm~\ref{alg:risk_prob_est}.

By adjusting the initial position and initial velocity of ego vehicle, we can get the safety probabilities of different states, which is useful for safe control design to be introduced in the next subsection.


\begin{algorithm}[t]
\caption{Determine Vehicle Safety}\label{alg:determine_safety}
\begin{algorithmic}[1]
\Procedure{egoSafety}{$T, p_0, v_0$} \Comment{time horizon, initial state, initial speed}
    \State \textbf{Given:} $\lambda$, collision\_dist, safe\_dist
    \State \textbf{Initialize} world
    \State \textbf{Initialize} ego\_vehicle = $e \leftarrow \text{init}(p_0, v_0)$ \Comment{initialize ego vehicle with initial position, speed}
    \State \textbf{Initialize} occlusion = $o$
    \State \textbf{Initialize} pedestrian = $w_i \leftarrow \text{Random}(\lambda_i), \forall i$ \Comment{initialize walker based on some random distribution parameterized by $\lambda$}
    
    \State \textbf{Initialize} target speed $v_{\text{target}}$ 
    \For {$t$ in $0:T$} \Comment{run until time horizon}
    \State emergency\_stop $\leftarrow$ false
    \State Determine visibility $c \leftarrow$ visibleWalker($e, w, o$)
    \If{$c =$ true}
    \State emergency\_stop = true
    \EndIf
    \If {emergency\_stop = true}
    \State Apply brakes to $e$
    \Else
    \State Calculate throttle 
    $u \leftarrow \text{NominalController}(e, v_{\text{target}})$
    \State Apply throttle $u$ to $e$
    \EndIf
    \If{$\sqrt{(e.p_x - w_i.p_x)^2 + (e.p_y - w_i.p_y)^2} <$ collision\_dist for some $i$} \Comment{collision distance}
    \State \textbf{return} $s \leftarrow 0$ \Comment{unsafe}
    \EndIf
    \If{$e.p_x$ $\geq$ safe\_dist} \Comment{safely past collision area}
    \State \textbf{return} $s \leftarrow 1$ \Comment{safe}
    \EndIf
    \EndFor
    \State \textbf{return} $s \leftarrow 1$ \Comment{time horizon reached}    
\EndProcedure
\end{algorithmic}
\end{algorithm}

\begin{algorithm}[t]
\caption{Risk Probability Estimation}\label{alg:risk_prob_est}
\begin{algorithmic}[1]
\Procedure{riskEstimate}{$T, v_0, p_0, N$} \Comment{time horizon, initial speed, initial position, number of samples}



    \For {$n \text{ in } 1:N$} \Comment{number of runs}

    \State 
    Determine safety $s_n \gets \text{egoSafety}({T, p_0, v_0})$

    \EndFor

    \State \textbf{return} safety probability $\safeprob = \sum_{n=1}^N \frac{s_n}{N}$


    
    
\EndProcedure
\end{algorithmic}
\end{algorithm}

\subsection{Safe Control}
\label{subsec:safe_control}
In this section, we propose the safe control strategy to ensure the long-term safety of the system. 

We consider the following discrete time kinematic model to approximate the vehicle dynamics
\begin{equation}
\label{eq:vehicle_dynamics}
\begin{aligned}
    p_{k+1} & = v_{k+1} \; \Delta t + p_k \\
    v_{k+1} & = u \; \Delta t + v_k,
\end{aligned}
\end{equation}
where $p$ is the position (in longitudinal direction) of the vehicle, $v$ is the longitudinal velocity, and $X^{\text{ego}} = [p, v]^\top$. 
Here we consider the implicit Euler method to simulate the dynamics, \ie the position dynamics is driven by the velocity at the next time step, otherwise Runge–Kutta or similar methods can be used. Here, the continuous time dynamics for the vehicle velocity is approximated by
\begin{equation}
\label{eq:continuous_dynamics}
    dv/dt = u.
\end{equation}
Then, we can write the long-term safety condition~\eqref{eqn:safe_cond} as
\begin{equation}
\label{eq:safe_constraint}
\begin{aligned}
    A\safeprob := \frac{d\safeprob}{dt} = & \frac{d\safeprob}{dv} \frac{dv}{dt} + \frac{d\safeprob}{dp} \frac{dp}{dt} \\
    = & \frac{d\safeprob}{dv}u + \frac{d\safeprob}{dp} v \geq -\alpha(\safeprob - (1-\epsilon)),
\end{aligned}
\end{equation}
where $\frac{d\safeprob}{dv}$, $\frac{d\safeprob}{dp}$, $v$, $\Delta t$, $\safeprob$ are either known or can be evaluated. Thus, \eqref{eq:safe_constraint} is a linear constraint on the control input $u$.

At each time step, we can solve for the safe control via the following constrained optimization
\begin{equation}
\label{eq:safe_control_optimization}
\begin{aligned}
    u^* & = \arg\min_{u \in \mathcal{U}} \|u - u_{\text{nominal}}\|_2^2 \\
    & \text{  s.t.  } \eqref{eq:safe_constraint}
\end{aligned}
\end{equation}
where $u_{\text{nominal}}$ is a reference nominal controller and will be used if the safety condition is satisfied, otherwise the safe control will be activated.
The procedures for safe control at occluded intersection is summarized in Algorithm~\ref{alg:safe_control}.

\begin{remark}
\label{rmk:ease of design}
The proposed optimization-based safe control~\eqref{eq:safe_control_optimization} is easy to design and implement. It only has function $\alpha$ and the desired risk tolerance $\epsilon$ as tunable parameters. It only imposes linear constraints on control, which can be optimized in quadratic program (QP) efficiently. \rev{The initial feasibility of~\eqref{eq:safe_control_optimization} is assumed as in Theorem~\ref{thm:1}. 
Relaxations of $\epsilon$ and $\alpha$ can be used to encourage feasibility when encountering ill-posedness issues or numerical errors.}
\end{remark}

\begin{algorithm}[t]
\caption{Probabilistic Safe Control}\label{alg:safe_control}
\begin{algorithmic}[1]
    \State \textbf{Given:} $N, \tau, v_0, p_0, \Delta x, \Delta v, v_{\text{target}}, T, d_{\text{min}}$, $\epsilon$ 

    \State $k \gets 0$

    \State $s \gets 1$ \Comment{safety indicator}

    \While {$k < T_{\text{end}}$ and $p_x < x_{\text{end}}$} \Comment{within simulation horizon and vehicle not passing the intersection}

    \State Get safety probability 
    
    $\safeprob = \text{risk\_estimate}(T, v_k, p_k, N)$

    \If{$\safeprob > 1 - \epsilon$}

    \State $u \gets \text{NominalController}(v_k, v_{\text{target}})$

    \Else

    \State Get safety probability at neighboring state 

    $\safeprob^\pm_p = \text{risk\_estimate}(T, v_k, p_k \pm \Delta x, N)$ 

    $\safeprob^\pm_v = \text{risk\_estimate}(T, v_k \pm \Delta v, p_k, N)$
    
    

    

    \State Estimate probability gradients 
    
    $\frac{d\safeprob}{dx} \approx \frac{\safeprob^+_p - \safeprob^-_p}{2\Delta x}$, 
    $\frac{d\safeprob}{dv} \approx \frac{\safeprob^+_v - \safeprob^-_v}{2\Delta v}$

    \State Solve for safe control $u$ through~\eqref{eq:safe_control_optimization}

    \EndIf

    \State Execute control $u$ and observe $p_{k+1}, v_{k+1}$

    \If{collision happens $\|p_{k+1} - p^{\text{ped}_i}_{k+1}\| < d_{\text{min}}$}

    \State $s \gets 0$

    \State Terminate the simulation

    \Else

    \State $k \gets k + 1$

    \EndIf

    \EndWhile

\State \textbf{return} $s$
\end{algorithmic}
\end{algorithm}

\section{Experiments}
\label{sec:experiments}
In this section, we demonstrate the efficacy of the proposed occlusion and interaction-aware safe controller using the CARLA autonomous driving simulator platform~\cite{dosovitskiy2017carla}. Code is available at \url{https://github.com/jacobwang925/AV_occlusion}.

\subsection{CARLA Setup}
\label{sec:carla_setup}


CARLA is an open-source simulator for development of autonomous driving systems. 
It has been widely used in autonomous driving research to simulate complex driving scenarios, such as collision avoidance~\cite{li2021risk}, lane keeping~\cite{samak2021proximally}, traffic sign recognition~\cite{mijic2021autonomous}, and end-to-end driving systems~\cite{xiao2020multimodal}.



For all experiments, we use CARLA version 0.9.10, with map Town03 and vehicle model Tesla Model 3. We run all simulations on a Linux 6.8.0 machine with one NVIDIA GeForce RTX 4090 GPU on Ubuntu 22.04.3.


The CARLA setting for the scenario of interest is shown in Fig.~\ref{fig:carla_setup}, where the ego vehicle is moving along the $x$ direction to pass the intersection. There is a parked truck at $(x, y) = (-7.0, 5.0)$, which is the visual occlusion in this scenario. Pedestrians come behind the truck starting from $(x, y) = (0, 13.0)$, and proceed to cross the intersection along the negative $y$ direction with constant speed $1 \mathrm{m/s}$. 
Pedestrians will come out at random times, and the waiting time for the occurrence of the pedestrian satisfies the following distribution.
For the first pedestrian, the waiting time $\Delta \tau_1$ is generated via a scaled truncated normal distribution through rejection sampling:
\begin{equation}
\label{eq:ped_distribution_1}
    \Delta \tau_1 \sim \mathcal{N}(1.5,6.25),\quad \Delta \tau_1 \in [0,10],
\end{equation}
and time interval $\Delta \tau$ between all subsequent pedestrians is generated via the following scaled truncated normal distribution:
\begin{equation}
\label{eq:ped_distribution_2}
    \Delta \tau \sim \mathcal{N}(6,6.25),\quad \Delta \tau \in [0,15].
\end{equation}

The nominal controller for the ego vehicle is a cruise controller, which maintains the same constant velocity as the initial velocity. The ego vehicle is also equipped with an emergency controller with constant brake force $0.05$.
The nominal control is executed by default, while the emergency control will be activated and will overwrite the nominal controller if the ego vehicle sees any crossing pedestrian. If the ego vehicle is not visually occluded, the condition for the ego vehicle to see any pedestrian in sight can be expressed mathematically as
\begin{equation}
    -10.0 < x_{\text{ego}} < 0.0
\end{equation}
\begin{equation}
    y_{\text{ped}}-6.5 < y_{\text{ego}} < y_{\text{ped}}+6.5
\end{equation}
where $(x_{\text{ego}}, y_{\text{ego}})$ is the position of the ego vehicle, and $y_{\text{ped}}$ is the $y$ position of the pedestrian. 
While some literature considers high-level trajectory planning, we apply the proposed technique to a low-level controller that limits the vehicle to maintain in the same lane for fixed $y_{\text{ego}} = 0$.


\begin{figure}[t]
    \centering
    \includegraphics[width=0.9\linewidth]{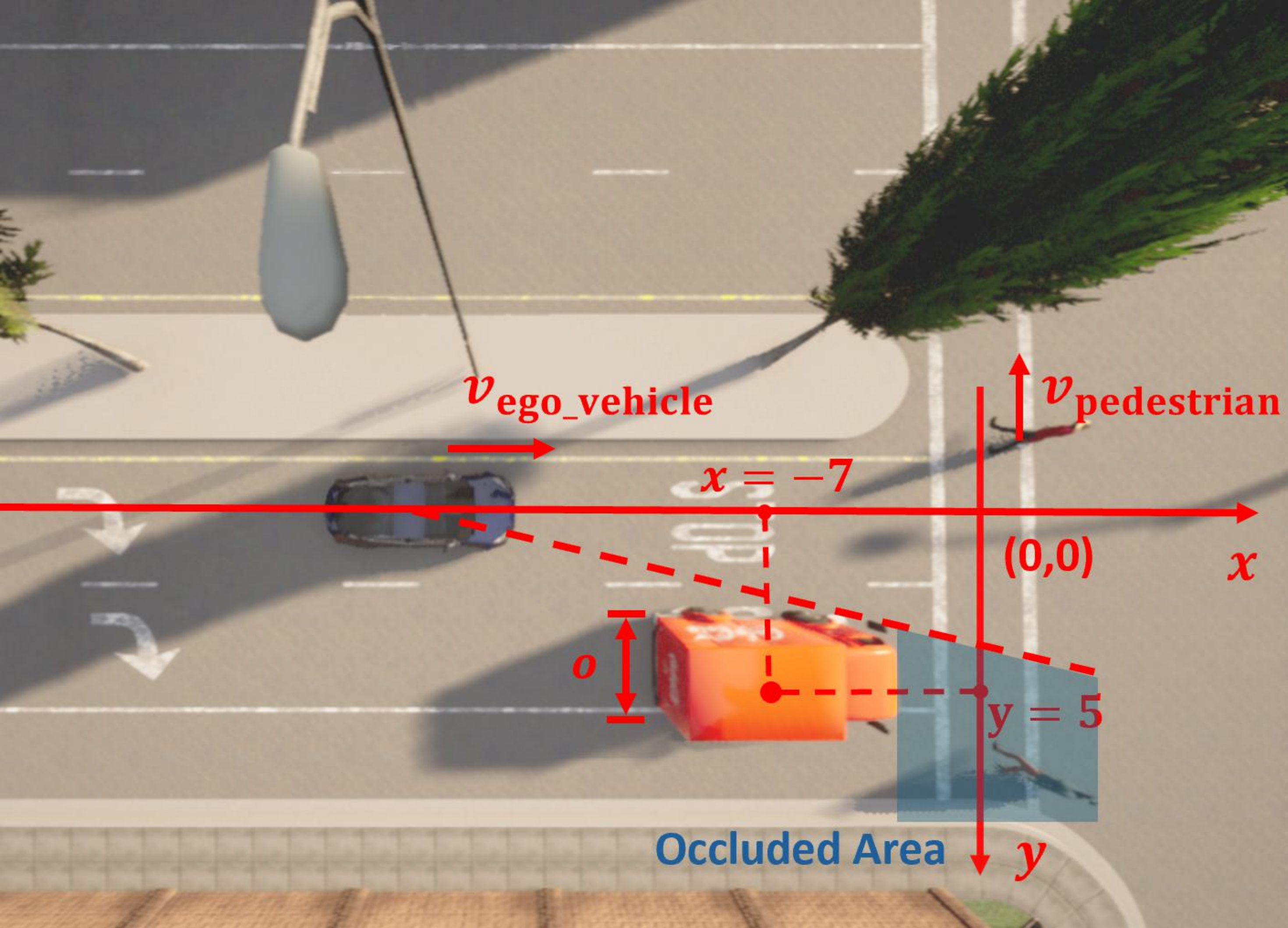}
    \caption{CARLA Scenario Setup.}
    \label{fig:carla_setup}
    \vspace{-0.5em}
\end{figure}

\subsection{Risk Probability Estimation}

We run Monte Carlo (MC) simulation to obtain safety probability of the nominal controller for time horizon $T=10 \mathrm{s}$ for different initial conditions (position and velocity).
We set the discrete time step for the simulator to be $\Delta t = 0.05 \mathrm{s}$. 
The parameters for the pedestrian model and the nominal controller is described in Section~\ref{sec:carla_setup}, and the MC computation scheme follows Alg.~\ref{alg:determine_safety} and Alg.~\ref{alg:risk_prob_est}.
Fig.~\ref{fig:heat_map_psi} shows the estimated safety probability with different initial states of the vehicle. We can see that in general the safety probability will increase as the vehicle gets away from the intersection and with slower initial speed. This matches the intuition that, with a slower speed and a farther away starting position, the vehicle can more easily come to a stop before the intersection with the emergency control to avoid a potential collision, resulting in a higher probability of safety.

\begin{figure}[t]
    \centering
    \includegraphics[width=0.9\linewidth]{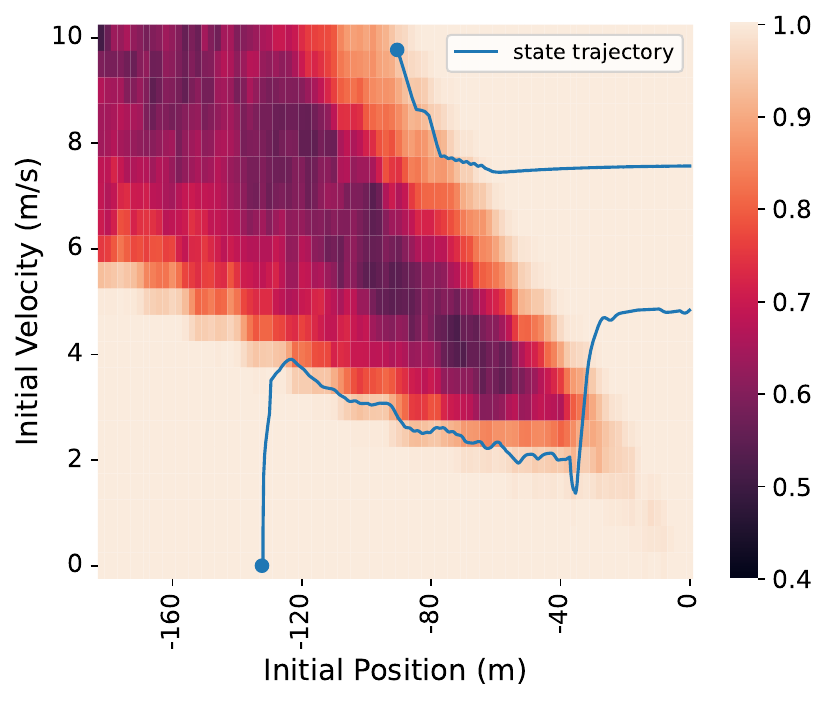}
    \caption{Estimation of safety probability $\mathbf{\Psi}(x)$ with grid discretization $dx_{\text{init}} = 2$ for initial position and $dv_{\text{init}} = 0.5$ for initial velocity. Blue curves show the typical vehicle trajectories under the proposed safe controller, with the blue dots as initial states.
    }
    \label{fig:heat_map_psi}
    \vspace{-0.5em}
\end{figure}

\subsection{Safe Control}


We run the safe control described by Algorithm~\ref{alg:safe_control}, where the safety probability is obtained by looking up the pre-computed table.
The safe control will drive the vehicle to move along the level set of the desired safety probability, unless the nominal controller is safe enough. Two sample trajectories of the safe control are visualized in blue lines in Fig.~\ref{fig:heat_map_psi}. 
It can be seen that for the trajectory below, when the ego vehicle is away from the intersection, the vehicle will slow down to maintain high safety probability. Once it passes a critical point, the vehicle will accelerate to pass the intersection safely. For the trajectory above, since the safety probability is often high enough near the initial state of the vehicle, it will maintain a target speed to pass through the intersection safely.

\vspace{-0.3em}

\subsection{Methods for Comparison}
We compare the proposed method with the following five baselines, where we limit our scope to methods with efficient online computation.

\textbf{Velocity tracking control (PID):} PID controllers~\cite{johnson2005pid} are implemented to track certain desired velocities for comparison.

\textbf{Worst-case control:} When there is latent risk detected (estimated risk probability less than 1), a constant braking force of $0.1$ will be applied for a duration of $0.25$ seconds. Note that this controller applies constant brake force and only accounts for the binary occurrence of risks, as opposed to the proposed method where the degree of the safety probability violation is considered. 

\textbf{Data-driven control (TransFuser):} We apply TransFuser~\cite{chitta2022transfuser}, a state-of-the-art data-driven controller that integrates multiple sensor inputs for safe and efficient autonomous driving to our scenario. \rev{The controller is trained on expert policy with additional sensors (such as LiDAR, IMU and depth camera) using privileged information (such as access to HD maps and ground truth traffic light states) from CARLA. Details about the method can be found in~\cite{chitta2022transfuser}, and we take the trained model published by the authors\footnote{https://github.com/autonomousvision/transfuser}.}

\rev{\textbf{Occlusion-aware model predictive control (OA-MPC):} We implement OA-MPC~\cite{firoozi2024oa}, a model predictive control method with reachability-based state constraints. The method specifies the maximum velocity of phantom pedestrians, and solve for the safe trajectory in the worst case scenario while minimizing costs. Specifically, we use~\eqref{eq:vehicle_dynamics} for the vehicle dynamics model and maximize the vehicle velocity as objective. The maximum pedestrian velocity is set to be $1 \mathrm{m/s}$ for the reachability analysis.}

\rev{\textbf{Planning-based control:} Following~\cite{zheng2025occlusion, moller2025shadows}, we implement a planning-based control that tracks a prespecified trajectory. The trajectory will decelerate and then stop at the intersection regardless of whether a pedestrian is visible in sight, before accelerating again to pass the intersection.}

\vspace{-0.3em}

\subsection{Results}
We run the proposed and baseline methods for the occluded intersection problem with initial position $x_\text{init} = -120 \mathrm{m}$ and initial velocity $v_\text{init} = 0 \mathrm{m/s}$.
Parameter $\alpha(h) = 0.2 h$ is used for the proposed method. 
Fig.~\ref{fig:vel_control} shows the typical velocity of the ego vehicle over the course.\footnote{The occurrence of high-frequency oscillation is not control method related, but rather due to the issue with the built-in vehicle dynamics in CARLA. Specifically, small brake commands in CARLA can cause rapid deceleration, which results in pedal commands later to compensate.} It can be seen that PID tracks a fixed velocity and does not take latent risk into accounts, resulting in collision. The data-driven method TransFuser is not stable and ends up with collision as well under this setting, possibly due to the distribution shift of the training and testing scenarios. \rev{The worst-case method, occlusion-aware MPC and the planning based method can drive the vehicle safely through the occluded intersection, but are overly conservative and have relatively slow velocity profiles.} In contrast, the proposed method travels through the intersection safely and with high efficiency.

\begin{figure}[t]
    \centering
    \includegraphics[width=0.85\linewidth]{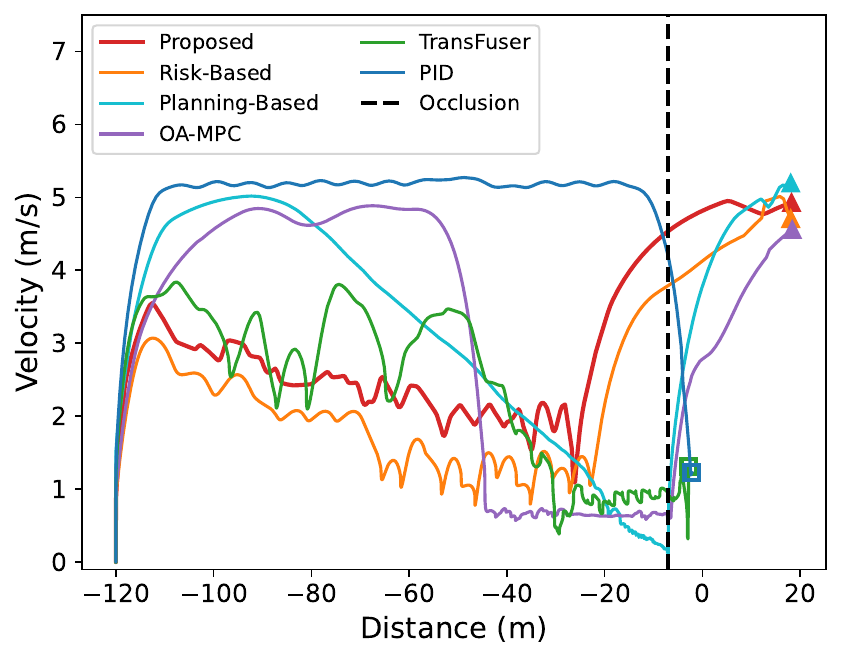}
    \caption{Vehicle velocities of the proposed safe controller and baselines methods. Triangle indicates safely pass through the intersection. Square indicates collision with pedestrian.}
    \label{fig:vel_control}
    \vspace{-0.3em}
\end{figure}

\begin{figure}[t]
    \centering
    \includegraphics[width=0.75\linewidth]{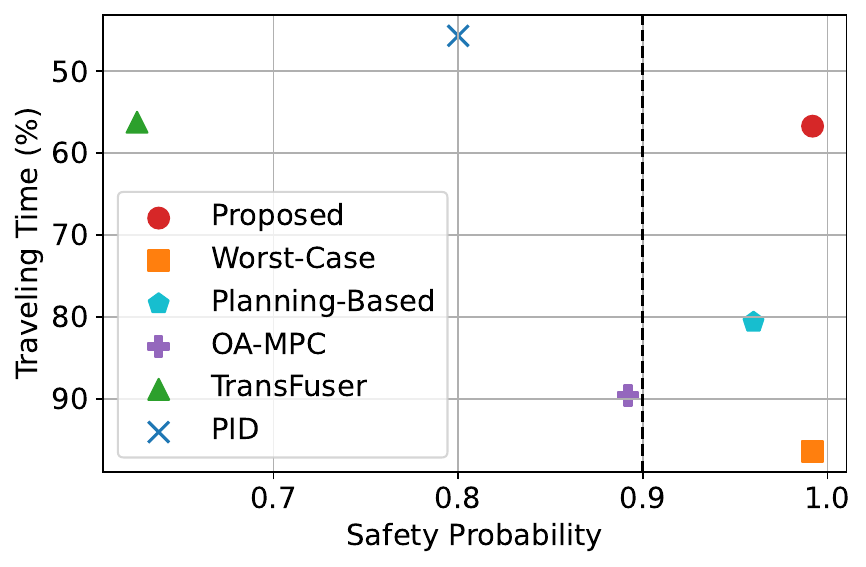}
    \caption{Safety-efficiency trade-offs.}
    \label{fig:trade_off}
    \vspace{-0.5em}
\end{figure}

Table~\ref{tab:emp_safe_prob} shows the empirical safety probability of the system and the corresponding traveling time under different initial state settings \rev{(see Table~\ref{tab:emp_safe_prob_full} in the Appendix for full results)}.\footnote{The dynamics model used for OA-MPC is not perfectly accurate which results in
collision in some cases.} The safety probability is calculated through 50 random trials.
It can be seen that with different risk tolerance $\epsilon$, the proposed method can ensure the empirical safety probability over the desired safety probability $1-\epsilon$. Among the methods that can produce the desired safety probability, the proposed method has the highest efficiency as the traveling time is lowest. Based on these data, a safety-efficiency trade-off plot is shown in Fig.~\ref{fig:trade_off}, where the two axes are the average safety probability and the average percentage of the maximum traveling time among all methods. It can be seen that the proposed method achieves the desired safety probability, and has short traveling time. \rev{The results indicate that the proposed method can ensure long-term safety under occlusion (against PID and data-driven methods), and has low traveling time thus high efficiency (against worst-case methods including OA-MPC and planning-based control).}

\begin{table}
  \caption{Empirical safety probabilities and traveling time. $\uparrow$ and $\downarrow$ indicate larger or smaller values preferred, respectively.}
  \label{tab:emp_safe_prob}
  \centering
  {
  \small
  \begin{tabular}{c|c|ccc}
    \hline
    Control Method & 
    Settings &
    $1-\epsilon$ & 
    $P_{\text{safe}}$ $\uparrow$  &
    $t \; (\mathrm{s})$ $\downarrow$\\
    \hline

    Proposed & & 0.9 & 0.98 & 26.94\\
    Worst-Case & & - & 1 & 55.49\\
    Planning-based & $x_{\text{init}}=-180\mathrm{m}$ & - & 0.96 & 34.64 \\
    OA-MPC         & $v_{\text{init}}=2\mathrm{m/s}$ & - & 0.94 & 38.29 \\
    TransFuser & & - & 0.53 & 33.76\\    
    PID &  & - & 1 & 32.68\\
    \hline

    Proposed & & 0.95 & 0.98 & 23.50\\
    Worst-Case & & - & 0.98 & 31.44\\
    Planning-based & $x_{\text{init}}=-120\mathrm{m}$ & - & 0.96 & 28.82 \\
    OA-MPC         & $v_{\text{init}}=6\mathrm{m/s}$ & - & 0.90 & 29.22 \\
    TransFuser & & - & 0.7 & 19.06\\
    PID &  & - & 0.74 & 9.34\\
    \hline

    Proposed & & 0.9 & 1 & 11.24\\
    Worst-Case & & - & 0.98 & 22.56\\
    Planning-based & $x_{\text{init}}=-60\mathrm{m}$ & - & 0.98 & 25.32 \\
    OA-MPC         & $v_{\text{init}}=2\mathrm{m/s}$ & - & 0.82 & 27.50 \\
    TransFuser & & - & 0.8 & 8.31\\
    PID &   & - & 0.9 & 14.94\\
    \hline
  \end{tabular}
  \vspace{-0.5em}
  }
\end{table}



\section{Conclusion}
\label{sec:conclusion}
This paper proposes an occlusion- and interaction-aware safe
control strategy that ensures long-term safety in the presence of latent risks without overly compromising performance. We demonstrate its reliability and computational efficiency via theoretical analysis and CARLA experiments. The results show that the proposed controller ensures long-term safety under occlusions and achieves better safety-performance trade-offs over existing worst-case and large data-driven methods. \rev{Future work includes designing and incorporating more efficient risk estimation techniques for safe control, conducting real-world experiments, and comparing the results with human driving behaviors.}

\bibliography{citation}

\begin{IEEEbiography}[{\includegraphics[width=1in,height=1.25in,clip,keepaspectratio]{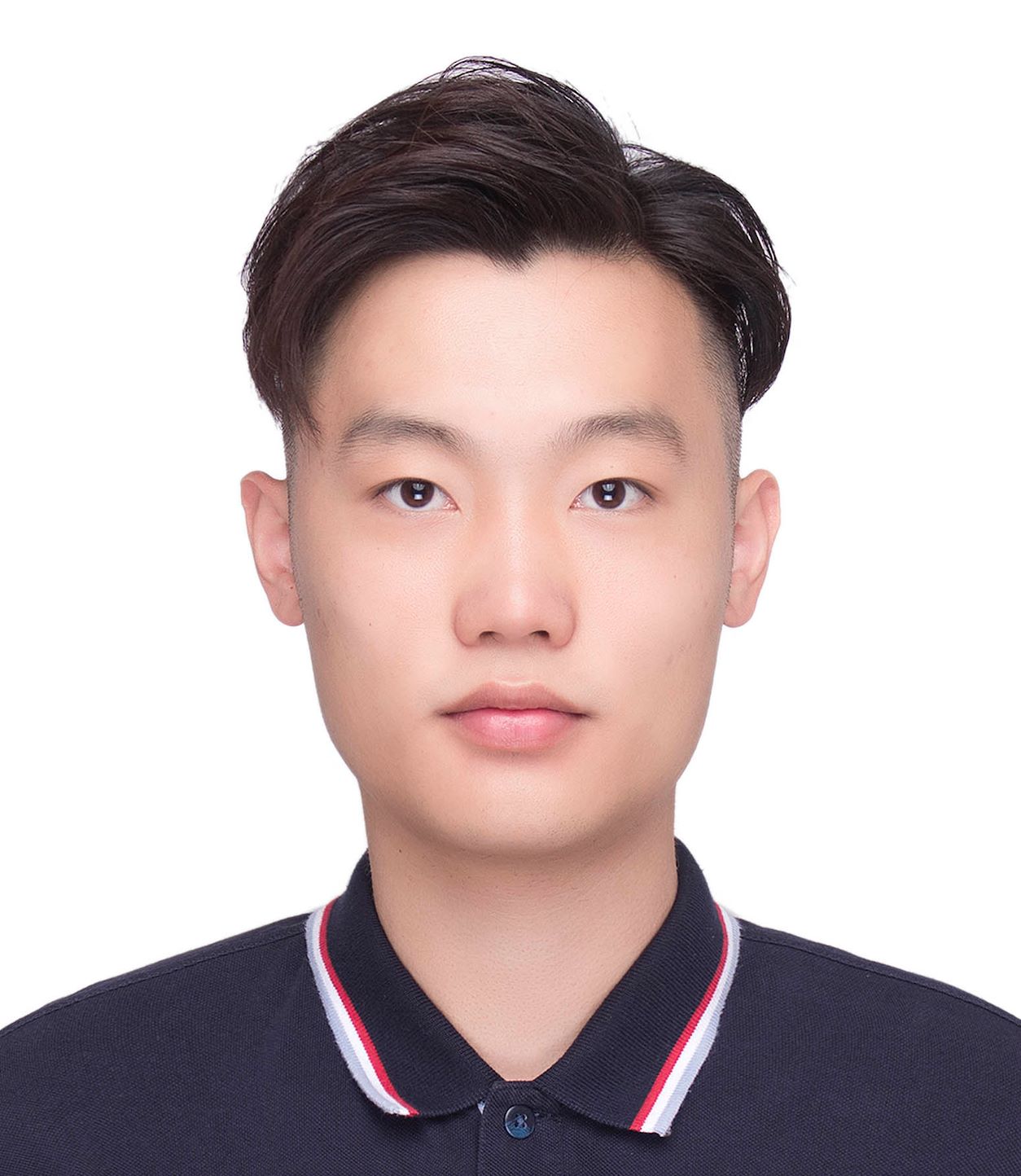}}]{Zhuoyuan Wang} received his B.E. degree in Automation from Tsinghua University, Beijing, China, in 2020 and his Ph.D. degree in Electrical and Computer Engineering from Carnegie Mellon University, Pittsburgh, PA, USA, in 2026.

His research interests include safety-critical control for stochastic systems, physics-informed learning, safe reinforcement learning and application to robotic systems.
He is a recipient of the Michel and Kathy Doreau Graduate Fellowship at Carnegie Mellon University.
\end{IEEEbiography}

\begin{IEEEbiography}[{\includegraphics[width=1in,height=1.25in,clip,keepaspectratio]{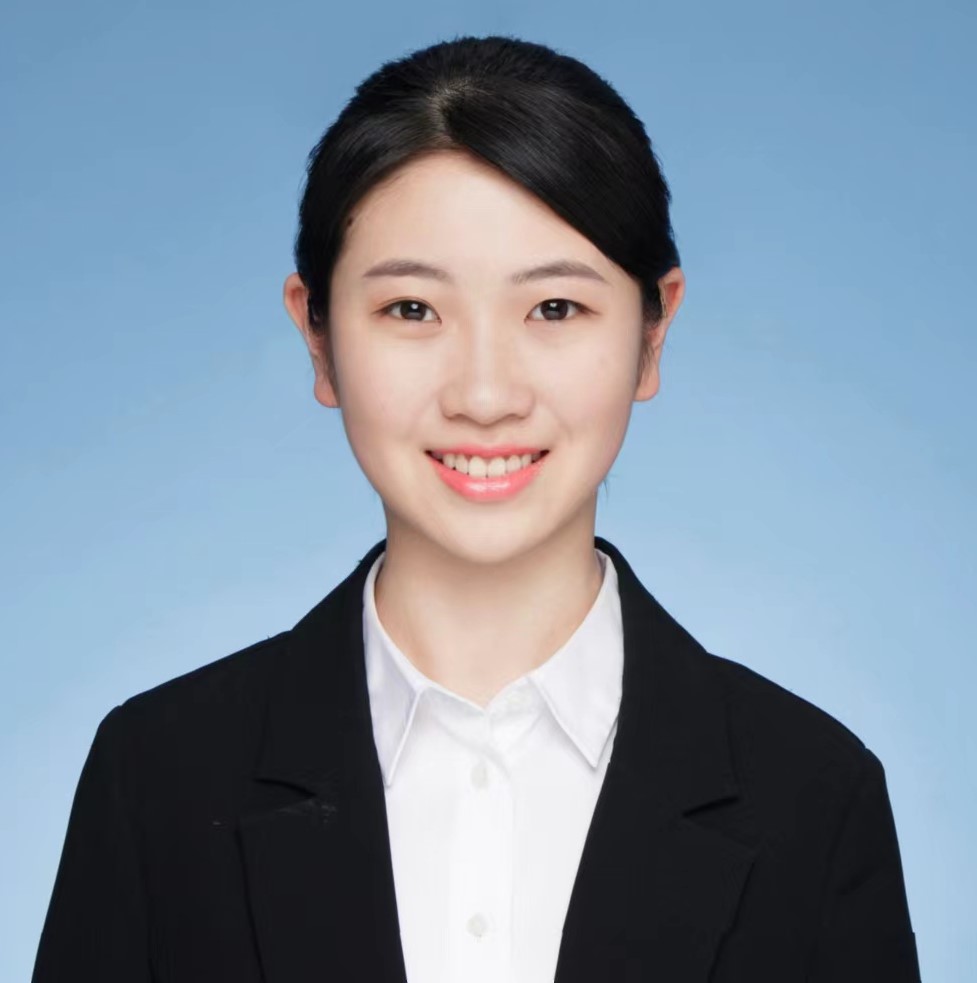}}]{Tongyao Jia} received her B.E. degree in Electronic Engineering from the Hong Kong University of Science and Technology, Hong Kong SAR, China, in 2023, and her Master's degree in 
Electrical and Computer Engineering from Carnegie Mellon University, Pittsburgh, PA, USA, in 2024.
\end{IEEEbiography}

\begin{IEEEbiography}[{\includegraphics[width=1in,height=1.25in,clip,keepaspectratio]{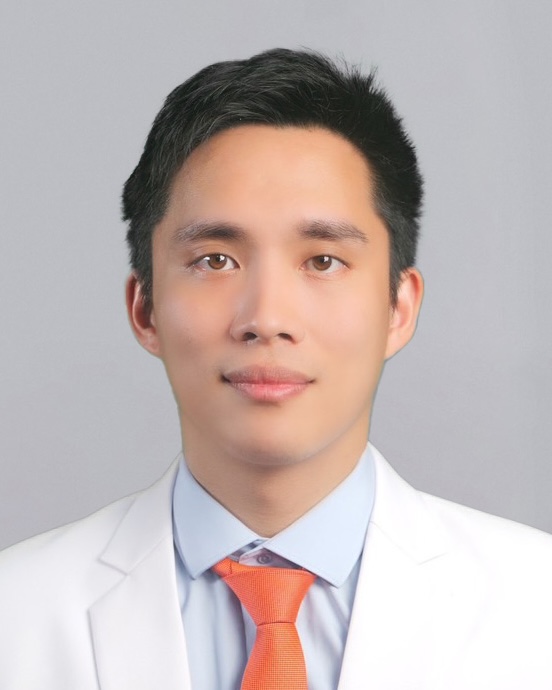}}] {Pharuj Rajborirug} received his B.S. degree from University of California Santa, Barbara, CA, USA, M.S.E. degree in Biomedical Engineering from Johns Hopkins University, MD, USA, and M.S. degree in Electrical and Computer Engineer at Carnegie Mellon University, PA, USA in 2024. He is currently a scientist at Faculty of Medicine, King Mongkut's Institute of Technology Ladkrabang (KMITL), Bangkok, Thailand.
\end{IEEEbiography}

\begin{IEEEbiography}[{\includegraphics[width=1in,height=1.25in,clip,keepaspectratio]{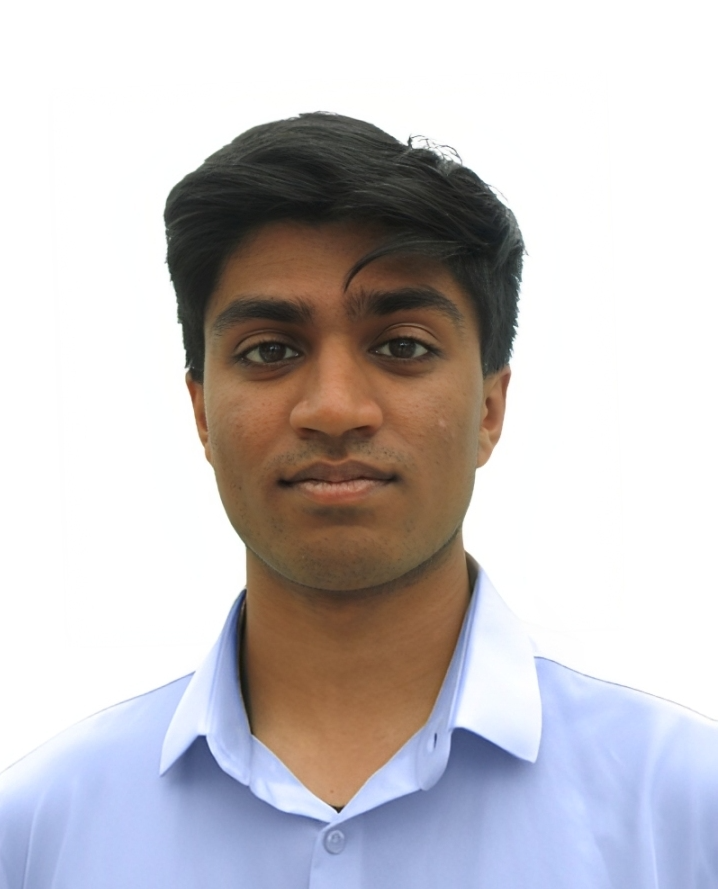}}] {Neeraj Ramesh} received his B.S. and M.S. degrees in Electrical and Computer Engineering from Carnegie Mellon University in Pittsburgh, PA, USA, with a concentration in Artificial Intelligence and Machine Learning. He is currently a Software Development Engineer on the Machine Learning Infrastructure team at Annapurna Labs in Cupertino, CA. His research interests include predictive computer vision, deep reinforcement learning for autonomous systems, and training and inference optimization for multimodal language models.
\end{IEEEbiography}

\begin{IEEEbiography}[{\includegraphics[width=1in,height=1.25in,clip,keepaspectratio]{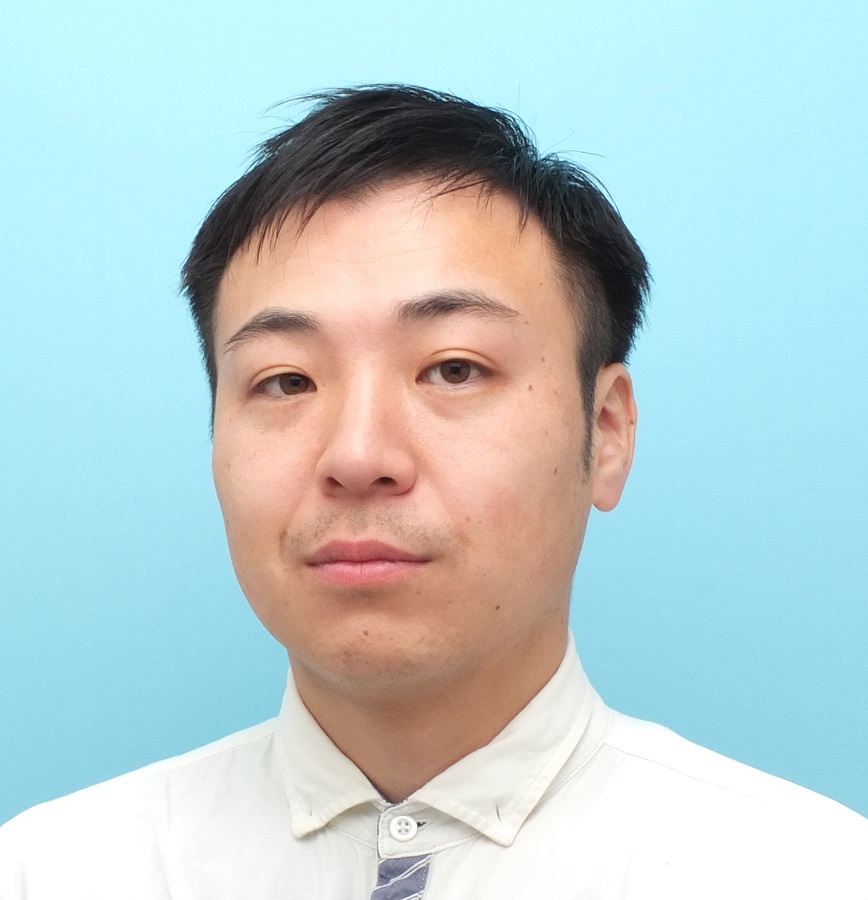}}]{Hiroyuki Okuda}  received the B.E. and M.E. degrees in Advanced Science and Technology from Toyota Technological Institute, JAPAN in 2005 and 2007, respectively, and the Ph.D. degree in Mechanical Science and Engineering from Nagoya University, JAPAN in 2010. 
He was a PD researcher with the CREST, JST from 2010 to 2012, an assistant professor of the Green Mobility Collaborative Research Center (GREMO) in Nagoya University from 2012 to 2016, and an assistant professor of the Department of Mechanical Science and Engineering in Nagoya University from 2017 to 2020. He was a visiting researcher of the Mechanical Engineering Department of U.C.Berkeley in 2018. Currently, he is an associate professor of the Graduate Department of Mechanical Systems Engineering of Nagoya University. His research interests are in the areas of system identification of hybrid dynamical system and its application to the modeling and the analysis of human behavior, the human-centered system design of autonomous/human-machine cooperative system. Dr.Okuda is a member of the IEEE, IEEJ, SICE, JSAE, and JSME.
\end{IEEEbiography}

\begin{IEEEbiography}[{\includegraphics[width=1in,height=1.25in,clip,keepaspectratio]{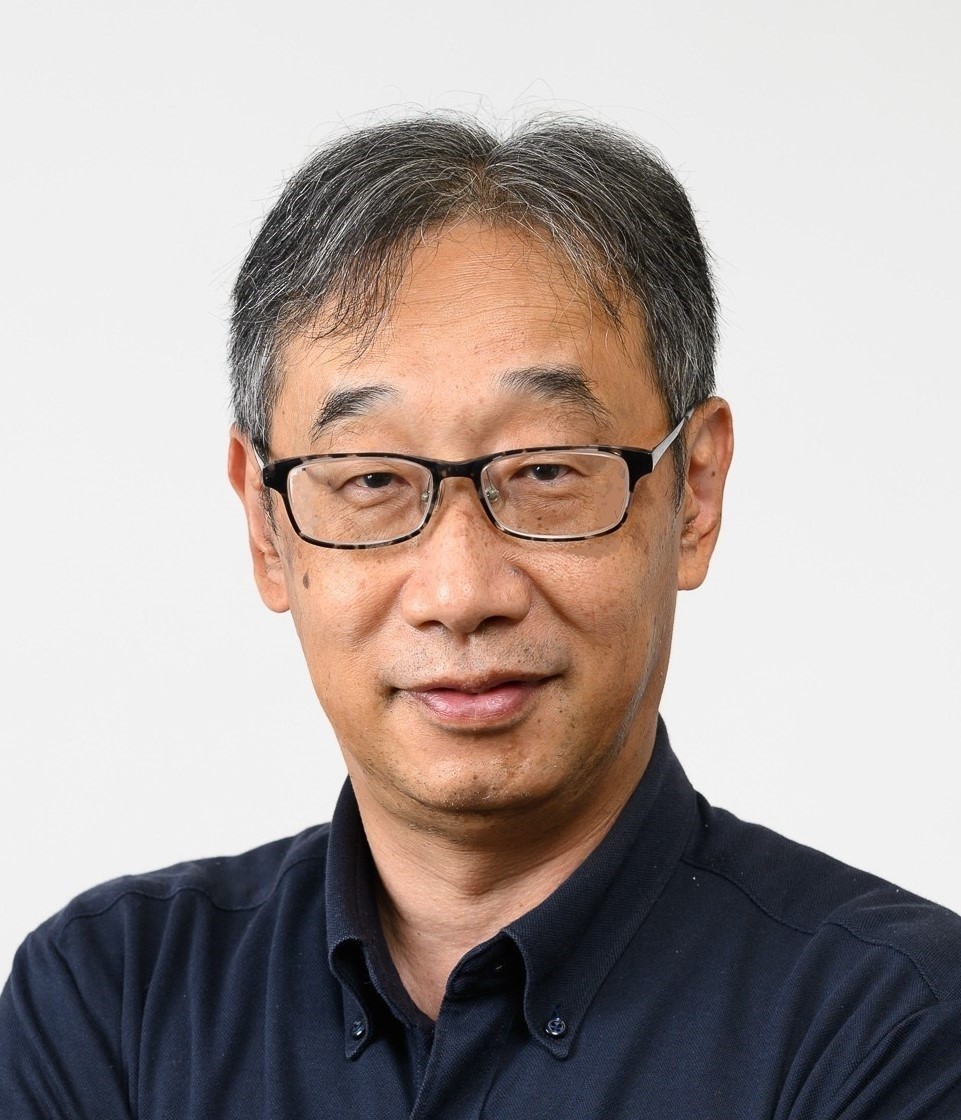}}]{Tatsuya Suzuki}  was born in Aichi, JAPAN, in 1964. He received the 
B.S., M.S. and Ph.D. degrees in Electronic Mechanical Engineering from 
Nagoya University, JAPAN in 1986, 1988 and 1991, respectively. 
From 1998 to 1999, he was a visiting researcher of the 
Mechanical Engineering Department of U.C.Berkeley. 
Currently, he is a Professor of the Department of Mechanical Systems Engineering, and Dean of the School of Engineering, Nagoya University.
He also has been an Executive Director of Global Research Institute for Mobility in Society (GREMO), Nagoya University in 2018-2020, and a Principal Investigator in JST, CREST in 2013-2019.
He won the best paper award in International Conference on Autonomic and Autonomous Systems 2017
and the outstanding paper award in International Conference on Control Automation 
and Systems 2008. He also won the journal paper award from IEEJ, SICE and JSAE. 
His current research interests are in the areas of analysis and design of 
human-centric intelligent mobility systems, and integrated design of transportation
and smart grid systems. 
\end{IEEEbiography}

\begin{IEEEbiography}[{\includegraphics[width=1in,height=1.25in,clip,keepaspectratio]{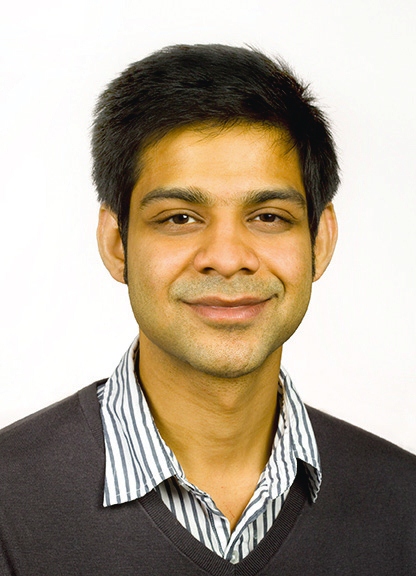}}] {Soummya Kar} received the B.Tech. degree in electronics and electrical communication engineering from the Indian Institute of Technology, Kharagpur, India, in May 2005 and the Ph.D. degree in electrical and computer engineering from Carnegie Mellon University in 2010. From June 2010 to May 2011, he was with the Electrical Engineering Department, Princeton University, as a postdoctoral research associate. He is currently the Buhl Professor of Electrical and Computer Engineering at Carnegie Mellon University. His research interests include decision making in large-scale networked systems, stochastic systems, multiagent systems and data science, with applications to cyberphysical and smart energy systems. He is a Fellow of the IEEE.
\end{IEEEbiography}

\begin{IEEEbiography}[{\includegraphics[width=1in,height=1.25in,clip,keepaspectratio]{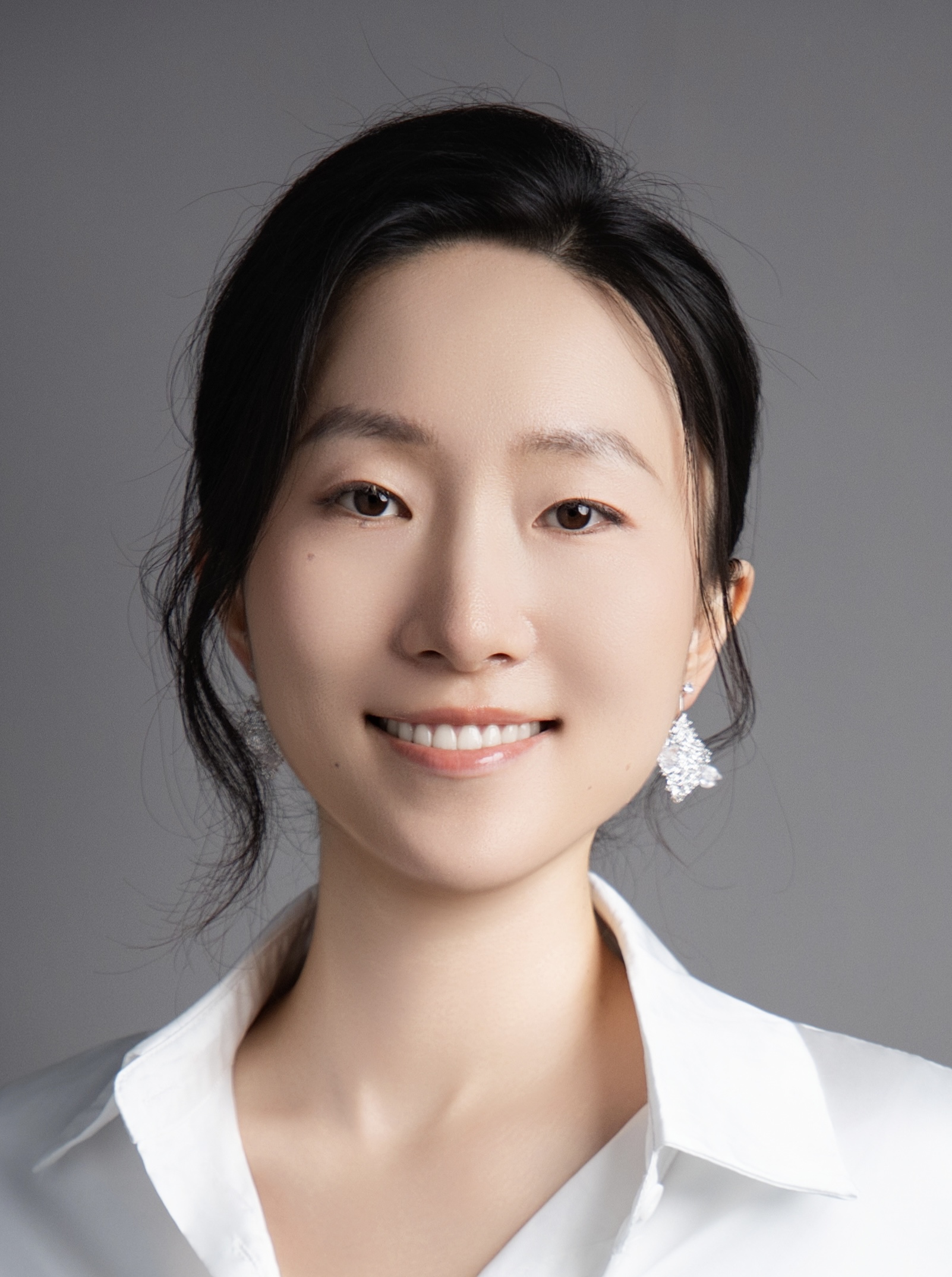}}]{Yorie Nakahira} is an Associate Professor in the Department of Electrical and Computer Engineering at Carnegie Mellon University. She received B.E. in Control and Systems Engineering from Tokyo Institute of Technology in 2012 and Ph.D. in Control and Dynamical Systems from California Institute of Technology in 2019. Her research interests include the fundamental theory of optimization, control, and learning and its application to neuroscience, cell biology, smart grid, cloud computing, finance, autonomous robots.
\end{IEEEbiography}


\appendix
\subsection{TransFuser Details}

TransFuser~\cite{chitta2022transfuser} is a multi‑modal fusion Transformer for end‑to‑end autonomous driving that integrates synchronized RGB images and LiDAR bird’s‑eye‑view (BEV) data within a unified attention‑based framework. \rev{Trained using imitation learning on trajectories generated by a privileged autopilot across 8 CARLA towns, it leverages a large-scale dataset of paired image and LiDAR sensor recordings along with expert driving actions. TransFuser employs multiple transformer modules at different feature resolutions to fuse perspective‑view image features and BEV LiDAR maps via self‑attention, enabling capture of long-range dependencies. In CARLA's public benchmarks, it achieves state‑of‑the‑art driving scores, outperforming geometry‑based baselines with 76\% in collision rate reduction~\cite{prakash2021multi}.}

In our experiments, we take the publicly released pre-trained TransFuser model.\footnote{https://github.com/autonomousvision/transfuser}
We visualize the scenario and the corresponding processed sensor inputs for TransFuser in Fig.~\ref{fig:transfuser_sensor}. Despite the additional use of LiDAR, IMU and depth camera, TransFuser fails to ensure safety of the system, likely due to the distribution shift of the training and testing scenarios. 

\begin{figure*}[t]
    \centering
    \includegraphics[width=0.99\linewidth]{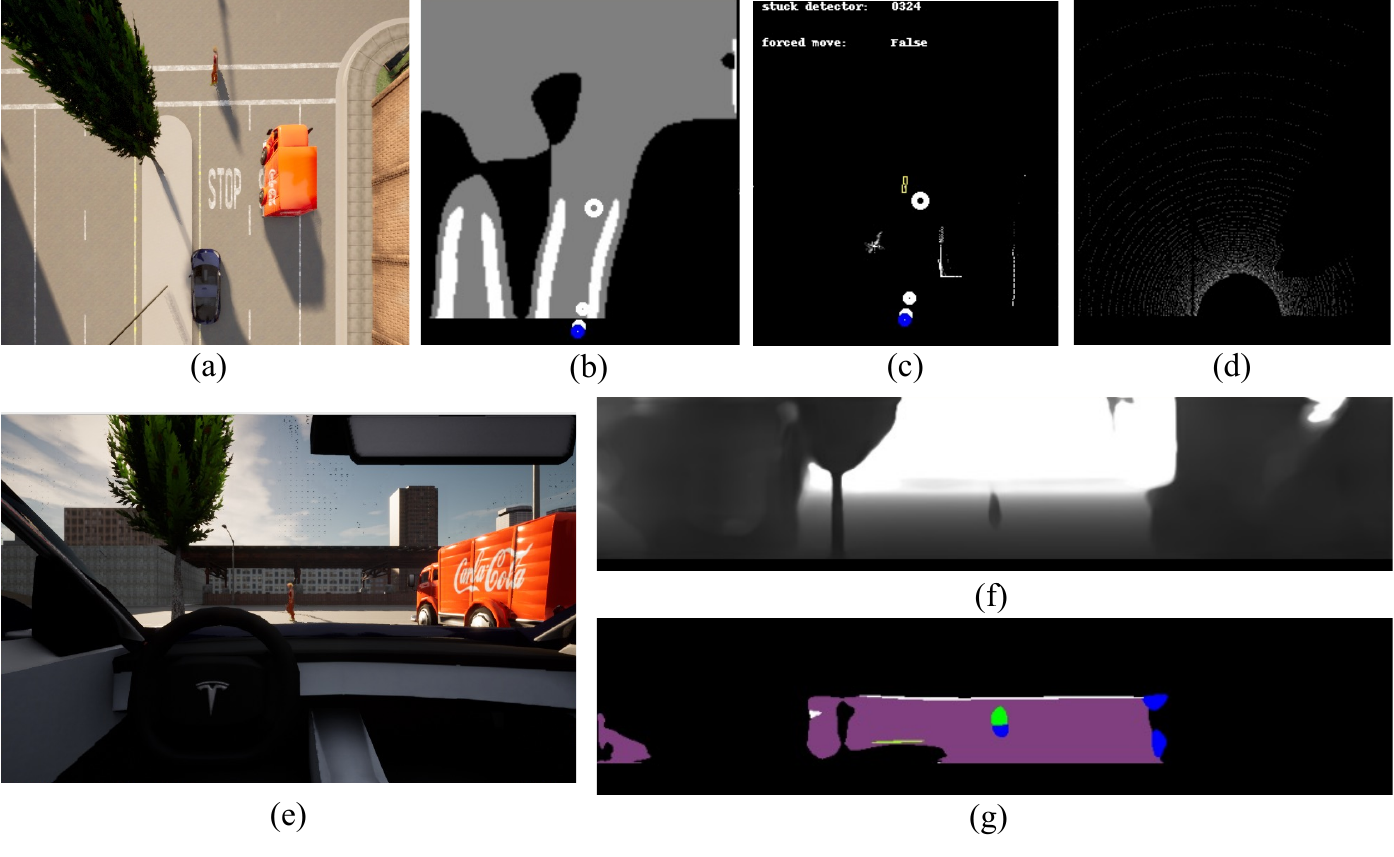}
    \caption{Visualization of the scenario and the corresponding processed sensor inputs for TransFuser~\cite{chitta2022transfuser}. (a) Birds' eye view of the scenario. (b) HD map prediction. (c) Obstacle prediction. (d) LiDAR ground panel. (e) First-person view from the ego vehicle. (f) Predicted depth. (g) Predicted semantics.}
    \label{fig:transfuser_sensor}
\end{figure*}


\subsection{Ablation Experiments}

\rev{We conduct ablation experiments on the effects of $\alpha$ for the proposed safe control scheme~\eqref{eqn:safe_cond}. We run the proposed controller with different $\alpha$ with initial position $x_\text{init} = -120 \mathrm{m}$ and initial velocity $v_\text{init} = 0 \mathrm{m/s}$. Fig.~\ref{fig:alpha_ablation} shows the vehicle velocity over distance and Table~\ref{tab:alpha_ablation} summarizes the statistics. It can be seen that with a higher coefficient value for $\alpha$, the safe controller becomes more aggressive for probabilistic safety condition enforcement, resulting in higher safety probability and slightly more oscillating trajectories. Nevertheless, the choice of $\alpha$ does not greatly affect the performance and all cases achieve the desired safety probability.}

\rev{We also show how the proposed method accounts for different occurrence probabilities of the pedestrian. Specifically, we consider the scenario in Section~\ref{sec:carla_setup} and consider two additional pedestrian distributions:
\begin{enumerate}
\item 
Distribution 1: The distribution~\eqref{eq:ped_distribution_1} and~\eqref{eq:ped_distribution_2}.
\item 
Distribution 2: Waiting time $\Delta \tau_1$ for the first pedestrian:
\begin{equation}
    \Delta \tau_1 \sim \mathcal{N}(2.5,13),\quad \Delta \tau_1 \in [0,10],
\end{equation}
and time interval $\Delta \tau$ between all subsequent pedestrians:
\begin{equation}
    \Delta \tau \sim \mathcal{N}(2.5,13),\quad \Delta \tau \in [0,15].
\end{equation}
\item 
Distribution 3: No pedestrian coming out behind the occlusion, \ie $\Delta \tau_1 = \infty$.
\end{enumerate}
Fig.~\ref{fig:distribution_velocity} shows the velocity of the ego vehicle with the proposed safe control scheme, under different pedestrian occurrence probabilities listed above. Fig.~\ref{fig:distribution_control} shows the corresponding control actions. It can be seen that for different pedestrian occurrence probabilities, the proposed method yields different safe control actions, thus different behaviors of the ego vehicle.}

\begin{figure}[t]
    \centering
    \includegraphics[width=0.85\linewidth]{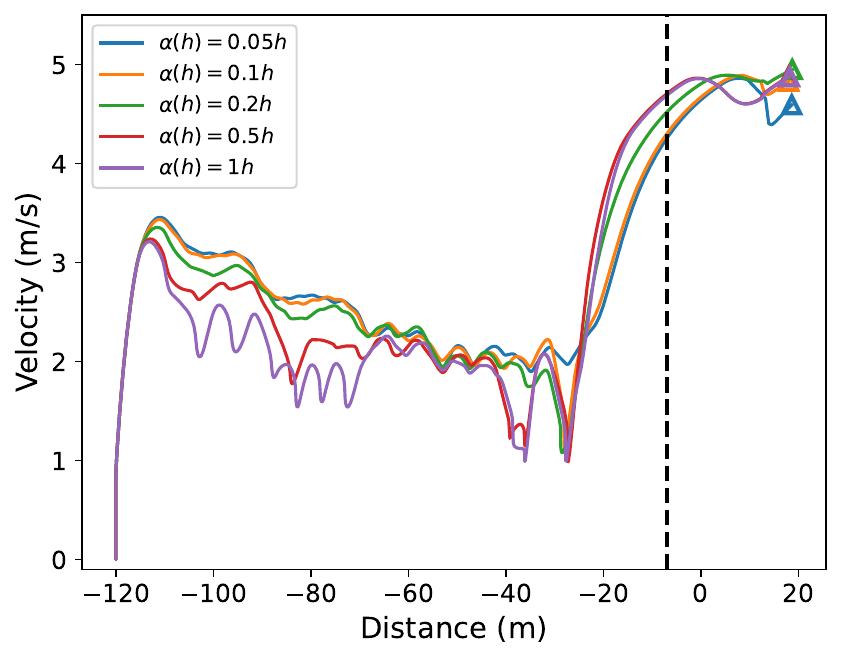}
    \caption{Vehicle velocities of the proposed safe controller with different $\alpha$. In all cases the vehicle safely pass through the intersection (indicated by the triangle).}
    \label{fig:alpha_ablation}
\end{figure}

\begin{table}
  \caption{Empirical safety probability and average traveling time with different $\alpha$. $\uparrow$ and $\downarrow$ indicate larger or smaller values preferred, respectively.}
  \label{tab:alpha_ablation}
  \centering
  {
  \small
  \begin{tabular}{c@{\hskip 1.5em}|@{\hskip 1.5em}c@{\hskip 2em}c}
    \hline
    $\alpha(h)$ & $P_{\text{safe}}$ $\uparrow$ &  $t \; (\mathrm{s})$ $\downarrow$ \\
    \hline
    $0.05h$ & 0.96 & 19.80 \\
    $0.10h$ & 1.00 & 20.95 \\
    $0.20h$ & 1.00 & 20.25 \\
    $0.50h$ & 1.00 & 22.05 \\
    $1.00h$ & 0.98 & 23.30 \\
    \hline
  \end{tabular}
  }
\end{table}

\begin{table}
  \caption{Empirical safety probabilities and traveling time. $\uparrow$ and $\downarrow$ indicate larger or smaller values preferred, respectively.}
  \label{tab:emp_safe_prob_full}
  \centering
  {
  \small
  \begin{tabular}{c|c|ccc}
    \hline
    Control Method & 
    Settings &
    $1-\epsilon$ & 
    $P_{\text{safe}}$ $\uparrow$  &
    $t \; (\mathrm{s})$ $\downarrow$\\
    \hline
    Proposed & & 0.95 & 1 & 25.36\\
    Worst-Case & & - & 1 & 46.98\\
    Planning-based & $x_{\text{init}}=-180\mathrm{m}$ & - & 0.94 & 32.98 \\
    OA-MPC         & $v_{\text{init}}=5\mathrm{m/s}$ & - & 0.86 & 41.60 \\
    TransFuser & & - & 0.58 & 31.8\\
    PID &  & - & 0.7 & 13.16\\
    \hline

    Proposed & & 0.9 & 0.98 & 26.94\\
    Worst-Case & & - & 1 & 55.49\\
    Planning-based & $x_{\text{init}}=-180\mathrm{m}$ & - & 0.96 & 34.64 \\
    OA-MPC         & $v_{\text{init}}=2\mathrm{m/s}$ & - & 0.94 & 38.29 \\
    TransFuser & & - & 0.53 & 33.76\\    
    PID &  & - & 1 & 32.68\\
    \hline

    Proposed & & 0.95 & 0.98 & 23.50\\
    Worst-Case & & - & 0.98 & 31.44\\
    Planning-based & $x_{\text{init}}=-120\mathrm{m}$ & - & 0.96 & 28.82 \\
    OA-MPC         & $v_{\text{init}}=6\mathrm{m/s}$ & - & 0.90 & 29.22 \\
    TransFuser & & - & 0.7 & 19.06\\
    PID &  & - & 0.74 & 9.34\\
    \hline

    Proposed & & 0.9 & 1 & 21.05\\
    Worst-Case & & - & 1 & 32.11\\
    Planning-based & $x_{\text{init}}=-120\mathrm{m}$ & - & 0.96 & 27.90 \\
    OA-MPC         & $v_{\text{init}}=3\mathrm{m/s}$ & - & 0.94 & 31.38 \\
    TransFuser & & - & 0.52 & 19.88\\
    PID &  & - & 0.66 & 18.50\\
    \hline

    Proposed & & 0.9 & 1 & 11.24\\
    Worst-Case & & - & 0.98 & 22.56\\
    Planning-based & $x_{\text{init}}=-60\mathrm{m}$ & - & 0.98 & 25.32 \\
    OA-MPC         & $v_{\text{init}}=2\mathrm{m/s}$ & - & 0.82 & 27.50 \\
    TransFuser & & - & 0.8 & 8.31\\
    PID &   & - & 0.9 & 14.94\\
    \hline
  \end{tabular}
  }
\end{table}

\begin{figure}
    \centering \includegraphics[width=0.85\linewidth]{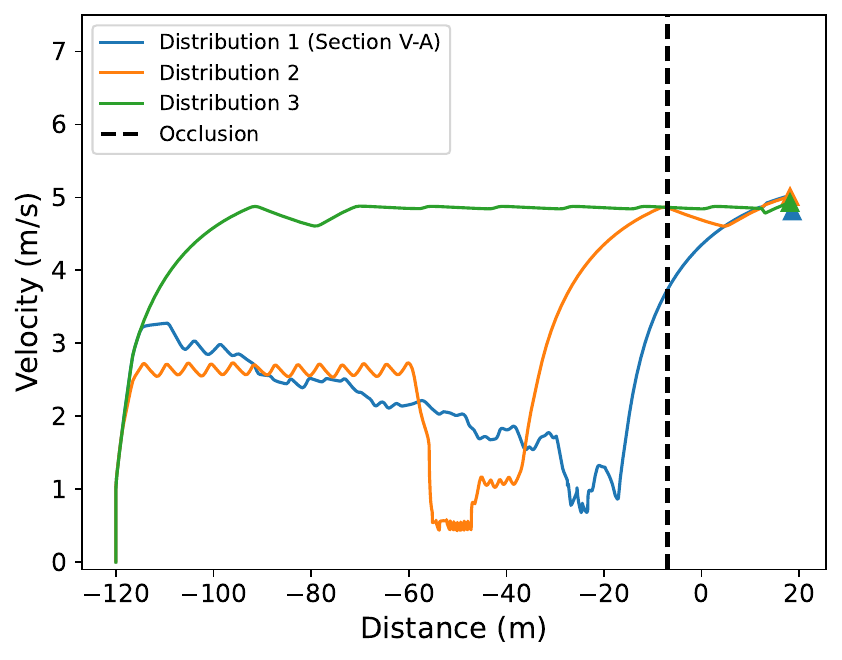}
    \caption{Ego vehicle velocity under different pedestrian occurrence probabilities. In all cases the vehicle safely pass through the intersection (indicated by the triangle).}
    \label{fig:distribution_velocity}
\end{figure}

\begin{figure}
    \centering \includegraphics[width=0.85\linewidth]{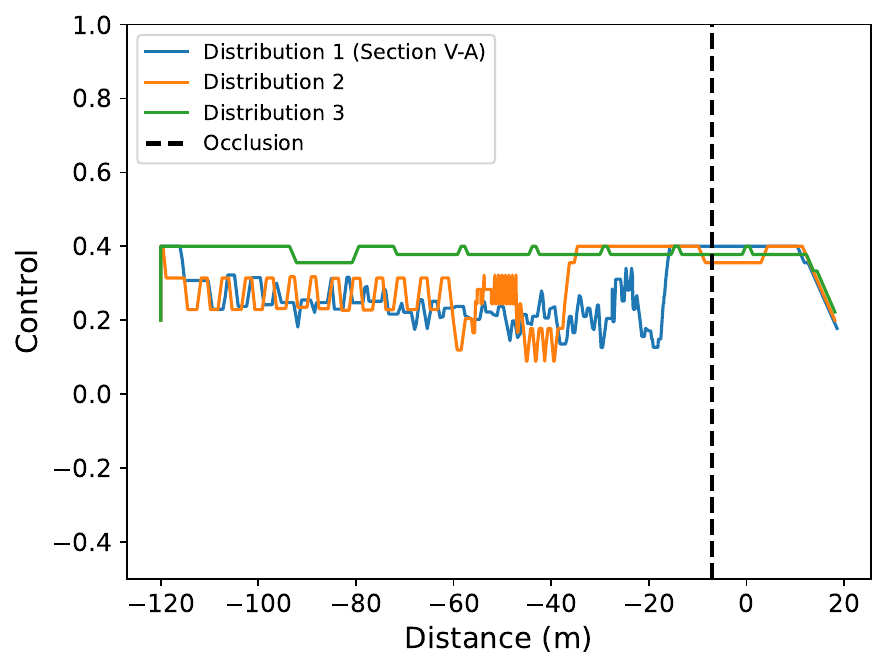}
    \caption{Safe control actions under different pedestrian occurrence probabilities.}
    \label{fig:distribution_control}
\end{figure}

\subsection{Supplementary Materials}

Video demonstrations comparing the proposed safe controller with baseline methods in the CARLA simulation environment are available at \url{https://youtu.be/7zrC2iWRmLg}.

As shown in the video, the PID controller tracks a fixed velocity until pedestrians become visible and emergency braking is activated. However, due to the high velocity when emergency braking is activated, the vehicle cannot come to a full stop before the intersection, resulting in a collision.

The worst-case controller regulates the vehicle velocity according to the estimated risk. However, since the level of risk is not explicitly considered, the resulting behavior is overly conservative. Moreover, this controller exhibits more oscillatory behavior, as it brakes for a fixed amount of time to maintain safety without accounting for the level of risk.

For the data-driven method TransFuser, the vehicle tends to slow down when encountering trees or street lights. This may be because the pre-trained TransFuser model interprets these objects as obstacles and reduces the vehicle speed to ensure safety. Such behavior is unnecessary and reduces efficiency in our scenario. Moreover, TransFuser results in a collision at the intersection, possibly due to the distribution shift between the training and testing scenarios.

Finally, the proposed controller regulates the vehicle velocity according to the level of risk, gradually decelerating as the vehicle approaches the intersection, resulting in safe trajectories while maintaining high efficiency.

\end{document}